\documentclass[preprint,12pt]{elsarticle}
\usepackage[T1]{fontenc}
\usepackage{amssymb,amsthm,amsmath}
\usepackage{tabularx}
\usepackage{xcolor}
\usepackage[small]{complexity}
\usepackage[skins,listings]{tcolorbox}
\usepackage{hyperref}
\usepackage{caption}
\usepackage[nameinlink,capitalise,noabbrev]{cleveref}
\usepackage{tikz}
\usetikzlibrary {graphs}
\usetikzlibrary {graphs.standard}
\usetikzlibrary{patterns.meta}
\hypersetup{colorlinks = true, citecolor = blue, linkcolor = blue, breaklinks = true, urlcolor=blue}
\Crefname{figure}{Fig.}{Figs.}
\Crefname{lemma}{Lemma}{Lemmas}
\newcommand{\kSAT}[1]{$#1$-$\SAT$}

\newcommand{\pname}[1]{\textcolor{black}{\textsc{#1}}}
\newtcolorbox[auto counter]{ProblemBox}[2][]{
	colframe=black!20!white,
	coltitle=black,
	arc=1.5mm,
	boxrule=.2mm,
	colbacktitle=white,
	colback=white,
	enhanced,
	adjusted title=flush left,
	attach boxed title to top center={yshift=-3mm,xshift=-3.5cm},
	boxed title style={colframe=white,righttitle=-3mm,lefttitle=-3mm},
	title=#2,#1
}

\newcommand{\problemClassic}[4]{
	\begin{ProblemBox}[label={#4},nameref={\pname{#1}}]{\pname{#1}}
		\begin{tabularx}{\textwidth}{lX}
			\textnormal{Input:} & \textnormal{#2} \\
                \textnormal{Question:} & \textnormal{#3}
		\end{tabularx}
		\vspace{-3,5mm}
	\end{ProblemBox}
}

\newtheorem{theorem}{Theorem}
\newtheorem{lemma}[theorem]{Lemma}
\newtheorem{corollary}[theorem]{Corollary}

\journal{Discrete Applied Mathematics}

\begin{document}

\makeatletter
\DeclareRobustCommand*{\nameref}{%
\color{black}%
        \@ifstar\T@nameref\T@nameref
        }%
\makeatother

\begin{frontmatter}

\title{The complexity of convexity number and percolation time in the cycle convexity}

\author[UFCA]{Carlos V.G.C. Lima}
\ead{vinicius.lima@ufca.edu.br}
\author[UFCA]{Thiago Marcilon\corref{cor1}}
\ead{thiago.marcilon@ufca.edu.br}
\author[UFC]{Pedro Paulo de Medeiros}
\ead{pedropmed1@gmail.com}
\affiliation[UFCA]{organization={CCT, Universidade Federal do Cariri},
            city={Juazeiro do Norte},
            country={Brazil}}
\affiliation[UFC]{organization={Departamento de Matemática, Universidade Federal do Ceará},
            city={Fortaleza},
            country={Brazil}}
\cortext[cor1]{Corresponding author}

\begin{abstract}
The subject of graph convexity is well explored in the literature, the so-called interval convexities above all. In this work, we explore the cycle convexity, an interval convexity whose interval function is $I(S) = S \cup \{u \mid G[S \cup \{u\}]$ has a cycle containing $u\}$. In this convexity, we prove that determine whether the convexity number of a graph $G$ is at least $k$ is \NP-complete and \W[1]-hard when parameterized by the size of the solution when $G$ is a thick spider, but polynomial when $G$ is an extended $P_4$-laden graph. We also prove that determining whether the percolation time of a graph is at least $k$ is \NP-complete even for fixed $k \geq 9$, but polynomial for cacti or for fixed $k\leq2$.
\end{abstract}


\end{frontmatter}

\section{Introduction}
\label{sec:intro}
We assume the reader possesses basic knowledge in Graph Theory and Complexity Theory, citing \cite{Bondy2008, Garey1979, cygan2015} as references for further exploration of these subjects. For a graph $G = (V,E)$, let $V(G) = V$, and $E(G) = E$. For any $v \in V(G)$ and $S \subseteq V(G)$, let $G[S]$ be the subgraph of $G$ induced by the subset of vertices $S$, $G-v = G[V \setminus \{v\}]$ and $G-S = G[V \setminus S]$. Henceforth, we consider only simple and finite graphs.

A \emph{convexity space}~\cite{bryant1972}, or simply \emph{convexity}, is an ordered pair $(V,\mathcal{C})$ where~$V$ is a set and $\mathcal{C}$ is a subset family of $V$, called \emph{convex sets}, satisfying the following properties: (\textbf{C1}) $V,\varnothing \in \mathcal{C}$; (\textbf{C2}) for all $\mathcal{C}' \subseteq \mathcal{C}$, we have $\bigcap \mathcal{C}'\subseteq \mathcal{C}$; and (\textbf{C3}) every nested union of convex sets is convex. When $V$ is finite, we can disregard property (\textbf{C3}). 

If $S\subseteq V$ is a convex set, we say that $V\setminus S$ is \emph{co-convex}. The \emph{convex hull} of a set $S \subseteq V$ concerning $(V,\mathcal{C})$ is the smallest convex set $C\in \mathcal{C}$ such that $C \supseteq S$, and we denote the convex hull of the set $S$ by $H(S)$. If $H(S)=V$, we say that $S$ is a \emph{hull set}. Note that if $S$ is a hull set and $C$ is co-convex, then $S \cap C \neq \varnothing$, since otherwise $S \subseteq V \setminus C$, which implies that $H(S) \subseteq V \setminus C \neq V$, as $V \setminus C$ is convex.

Convexity spaces generalize the notion of convexity in Euclidean space, making it applicable to a broader range of mathematical structures. In this paper, we exclusively refer to convexity spaces $(V,\mathcal{C})$, where $V$ denotes the vertex set of some finite graph, which are called \emph{graph convexities}. 

Let $G=(V,E)$ be a graph and~$S \subseteq V$. Many graph convexities $(V,\mathcal{C})$ are defined by an \emph{interval function} $I : 2^V \rightarrow 2^V$, where $I^0(S) = S$ and $I^k(S) = I(I^{k-1}(S))$, for all $k \geq 1$. Thus, the convex hull of $S$ is defined as $H(S) = I^{\ell}(S)$, for some~$\ell \geq 0$, and such that $I^k(S) = I^{k+1}(S)$, for all~$k \geq \ell$. In this case, we call such a graph convexity as an \emph{interval convexity}. By this definition, the convex hull of a set $S$ can be obtained by applying the interval function iteratively to $S$ until the resulting set no longer expands. As a consequence of this, in an interval convexity with interval function $I$, $S$ is a convex set if and only if $I(S) = S$. Henceforth, if $S$ is a convex set/hull set in a convexity on a graph $G$, we simply say that $S$ is a convex set/hull set of $G$, since the convexity is always clear from context. 

Typically, in interval convexities, $I(S)$ contains all vertices in $S$ and all vertices in any path with endpoints in $S$ that have some specific property. Three of the most studied interval convexities are the \emph{$P_3$ convexity}~\cite{bollobas2006,centeno2010}, the \emph{geodesic convexity}~\cite{pelayo2013,everet} and the \emph{monophonic convexity}~\cite{du1988,faja}, where~$I(S)$ contains, in addition to all vertices in $S$, all vertices in every $P_3$ with endpoints in $S$, all vertices in every shortest path with endpoints in $S$, and all vertices in every induced path with endpoints in $S$, respectively.

In this paper, we work with the recently defined \emph{cycle convexity}~\cite{araujo2018a}. The cycle convexity is an interval convexity such that $I(S) = S \cup \{u \mid G[S \cup \{u\}]$ has a cycle containing $u\}$. According to Araujo \emph{et al.}~\cite{araujo2018a}, the main motivation behind the conception of the cycle convexity stemmed from its application in the examination of the tunnel number of a knot or link within Knot Theory.

In~\cite{araujo2018a}, the \emph{interval number} of a graph in the cycle convexity is investigated. The interval number of a graph $G$, denoted by $in(G)$, is the size of a smallest hull set $S$ of $G$, such that $I(S) = V(G)$. The authors proved some bounds for $in(G)$, also exploring this parameter on some types of grids. They showed that the problem to determine whether $in(G) \leq k$ is \NP-complete for split graphs or bounded-degree planar graphs and \W[2]-hard for bipartite graphs when the problem is parameterized by $k$. On the other hand, they showed that the problem is polynomial for outerplanar graphs, cobipartite graphs and interval graphs, and is $\FPT$ when parameterized by the treewidth or the neighborhood diversity of $G$ or by $q$ in $(q,q-4)$-graphs.

In~\cite{araujo2024}, Araujo \emph{et al.} examined the \emph{hull number} in the cycle convexity. The hull number of a graph $G$, denoted by $hn(G)$, is the size of a smallest hull set $S$ of $G$. They showed that $hn(G) \leq n/2$ when $G$ is a $4$-regular planar graph. Furthermore, they proved that the problem to determine whether $hn(G) \leq k$ is \NP-complete, even when $G$ is a planar graph, but it is polynomial when $G$ is a chordal graph, $P_4$-sparse, or a grid.




In this paper, we explore the convexity number and the percolation time of a graph in the cycle convexity. The \emph{convexity number}~\cite{Chartrand1999, Chartrand2002} of a graph~$G=(V,E)$, denoted by $con(G)$, is the cardinality of a largest convex set of $G$ excluding~$V$ itself. The \emph{percolation time}~\cite{benevides2015b,benevides2015} of a graph $G$, denoted by~$pn(G)$, is defined as the largest integer $k > 0$ such that there exists a hull set $S$ of $G$ where $I^{k-1}(S) \neq V$ if such an integer $k$ exists; otherwise, $pn(G) = 0$. In \Cref{sec:convexnumber}, we show that determining whether $con(G) \geq k$ is \NP-complete and \W[1]-hard when parameterized by $k$, even when $G$ is a thick spider, but it is polynomial when $G$ is an extended $P_4$-laden graph. In \Cref{sec:percolationtime}, we establish that determining whether $pn(G) \geq k$ is \NP-complete for any fixed~$k \geq 9$, but it is polynomial for any fixed $k \leq 2$, or when $G$ is a cactus. In \Cref{sec:finalremarks}, we present our conclusions and final considerations.

\section{Results on the Convexity Number}
\label{sec:convexnumber}

In this section, we present our results regarding the convexity number in the cycle convexity. We determine the convexity number of $G$, when $G$ is the union/join of two graphs, a pseudo-split graph, and a quasi-spider. With that, we are able to obtain the two main results of this section. We prove that the following problem is \NP-complete and \W[1]-hard when parameterized by~$k$ for thick spiders and polynomial for extended $P_4$-laden graphs.

\problemClassic{Convexity number}{A graph~$G=(V,E)$ and an integer~$k>0$.}{$con(G) \geq k$?}{prob:convexityNumber}

Denote by $\alpha(G)$ and $\delta(G)$ the size of a largest independent set and the minimum degree of $G$, respectively. Let the \emph{union} of two graphs $G_1 = (V_1,E_1)$ and $G_2 = (V_2,E_2)$ be $G_1 \cup G_2 = (V_1 \cup V_2, E_1 \cup E_2)$, and the \emph{join} of~$G_1$ and $G_2$ be $G_1 \lor G_2 = (V_1 \cup V_2, E_1 \cup E_2 \cup \{uv \mid u \in V_1, v \in V_2\})$. 
\begin{lemma}
\label{lemma:convexityNumberunionjoin}
Let $G_1 = (V_1,E_1)$ and $G_2 = (V_2,E_2)$ be two graphs and $c_i$ be the size of the smallest connected component of $G_i$, for $i \in \{1,2\}$. Then, $con(G_1 \cup G_2) = \max(|V_1| + con(G_2),con(G_1) + |V_2|)$ and
\begin{flalign*}
con(G_1 \lor G_2) = \begin{cases}
\max(\alpha(G_1),\alpha(G_2))&\text{if }|V_1|,|V_2| \geq 2\\
\max(\alpha(G_2),|V_2|-c_2+1)&\text{if }|V_1| =1\\
\max(\alpha(G_1),|V_1|-c_1+1)&\text{if }|V_2| =1.
\end{cases}
\end{flalign*}
\end{lemma}

\begin{proof}
Let $T_i$ be a largest convex set of $G_i$ different from $V_i$, for $i \in \{1,2\}$. Both sets $T_1 \cup V_2$ and $T_2 \cup V_1$ are convex sets of $G_1 \cup G_2$ and different from~$V_1 \cup V_2$. Thus, $con(G_1 \cup G_2) \geq \max(|V_1| + con(G_2),con(G_1) + |V_2|)$. Now, let $S$ be a set of vertices of size at least $1 + \max(|V_1| + con(G_2),con(G_1) + |V_2|)$, $S_1 = S \cap V_1$, and $S_2 = S \cap V_2$. Since $|S| > |V_1| + con(G_2)$ and $|S| > con(G_1) + |V_2|$, we have $|S_1| > con(G_1)$ and $|S_2| > con(G_2)$. Thus, $H(S) = H(S_1 \cup S_2) \supseteq H(S_1) \cup H(S_2) = V_1 \cup V_2 = V(G)$. Thus, we have that $S$ is either not convex or equal to $V(G)$. As a result, $con(G_1 \cup G_2) = \max(|V_1| + con(G_2),con(G_1) + |V_2|)$.

Let $G = G_1 \lor G_2$. Assume that $V_1 = \{v\}$ and let $C$ be the smallest connected component of $G_2$. A largest independent set of $G_2$ is also a convex set of $G$ and is different from $V(G)$ as it certainly does not include $v$. Furthermore, no vertex $u \in V(C)$ is in a cycle in $G[\{u\} \cup (V(G) \setminus V(C))]$, since $C$ is a connected component of $G - v$. Thus, the set $T = V(G) \setminus V(C) \neq V(G)$ is a convex set of $G$. Hence, we conclude that $con(G_1 \lor G_2) \geq \max(\alpha(G_2),|V_2|+1-c_2)$. Now, let $T$ be a set of vertices of $G$ such that $|T| \geq 1 + \max(\alpha(G_2),|V_2|+1-c_2)$. Since $|T| \geq 1 + \alpha(G_2)$, if $v \notin T$, $T$ is not an independent set of $G_2$ and thus $v \in I(T)$. Furthermore, since $|T| \geq 1 + (|V_2|+1-c_2)$ and $c_2$ is the size of $C$, which is the smallest connected component of $G_2$, $T$ has at least one vertex in each connected component of $G_2$. Therefore, since $v \in H(T)$, $H(T) = V(G)$, which implies that $T$ is not convex or equal to $V(G)$. We can then conclude that $con(G_1 \lor G_2) = \max(\alpha(G_2),|V_2|+1-c_2)$ if $|V_1| = 1$. The case where $|V_2| = 1$ is symmetrical to this case.

Now, assume that $|V_1|\geq 2$ and $|V_2| \geq 2$. A largest independent set of $G_1$ or $G_2$ is also an independent set of $G = G_1 \lor G_2$ and consequently a convex set different from $V(G)$. Therefore, $con(G_1 \lor G_2) \geq \max(\alpha(G_1),\alpha(G_2))$. 

Note that if $v \in V_1$ and $u,w \in V_2$, then $I(\{v,u,w\}) \supseteq V_1$, since any vertex in $V_1$ forms a $C_4$ with $v,u$ and $w$. This implies that $I^2(\{v,u,w\}) = V(G)$, since $|V_1| > 1$ and any vertex in $V_2 \setminus \{w\}$ forms a $C_4$ with $u$, $v$, and any other vertex in $V_1$. Similarly, this also happens if $v \in V_2$ and $u,w \in V_1$.

Let $T$ be a set of vertices of $G$ such that $|T| \geq 1 + \max(\alpha(G_1),\alpha(G_2))$. Let $T_i = T \cap V_i$ for $i \in \{1,2\}$. If $T_1$ is not an independent set of $G_1$, let $uw \in E_1$ such that $u,w \in T_1$. Then, $I(T_1) \supseteq I(\{u,w\}) \supseteq V_2$, which implies that $I^2(T_1) \supseteq V(G)$, since $u,w \in V_1$ and $|V_2| \geq 2$. It follows that $H(T) = V(G)$. So, $T$ is not convex or $T = V(G)$. With symmetrical arguments, we can reach the same conclusion when $T_2$ is not an independent set of $G_2$. 

Let $T_i$ be an independent set of $G_i$ for $i \in \{1,2\}$. Let $i \neq j \in \{1,2\}$. Since $|T| \geq 1 + \alpha(G_i)$, $|T_j| \geq 1$. Furthermore, if $T_i$ is not a largest independent set of $G_i$, then $|T_j| \geq 2$. This implies that $H(T) = V(G)$, since $|T_i| \geq 1$, and thus $T$ is not convex or $T = V(G)$. So, $T_1$ and $T_2$ are largest independent sets of $G_1$ and $G_2$ respectively. Also, if either $|T_1| \geq 2$ or $|T_2| \geq 2$, then either~$T$ is not convex or $T = V(G)$, since $|T_1|,|T_2| \geq 1$.

Therefore, assume that $T_1$ and $T_2$ are largest independent sets of $G_1$ and~$G_2$ respectively, and $|T_1|=|T_2|=1$, from which we can infer that $G_1$ and $G_2$ are complete graphs. Thus, $G = G_1 \lor G_2$ is a complete graph itself, which implies that $H(T) = V(G)$, since $|T|=|T_1|+|T_2|=2$ and $G$ is a complete graph. Hence, we have that $T$ is not convex or $T = V(G)$.

Therefore, we can conclude that $con(G_1 \lor G_2) = \max(\alpha(G_1),\alpha(G_2))$ if~$|V_1|,|V_2| \geq 2$. 
\end{proof}

A useful lemma about the convexity number of some graphs is the following:
\begin{lemma}
\label{lemma:conv1}
Let $G$ be a graph. We have that $con(G) = |V(G)|-1$ if and only if there is a vertex $v \in V(G)$ that is in no cycle of $G$.
\end{lemma}

\begin{proof}
If there is a vertex $v$ that is in no cycle of $G$, then $S = V(G) \setminus \{v\}$ is convex since $H(S) = S$. Therefore, $con(G) = |V(G)| - 1$.

On the other hand, assume that every vertex of $G$ is in at least one cycle of $G$. Let $v$ be an arbitrary vertex of $G$ and $S = V(G) \setminus \{v\}$. Since $v$ is in a cycle of $G$, $H(S) = V(G)$. Therefore, we conclude that there is no convex set with size $|V(G)| - 1$, which implies $con(G) \neq |V(G)| - 1$.
\end{proof}

A \emph{pseudo-split} graph is a graph $G = (V,E)$ whose vertex set can be partitioned into three subsets $S,C$ and $R$ such that (i) $|C|, |S| \geq 2$, $S$ is an independent set, and $C$ is a clique; (ii) there are no edges with one endpoint in $S$ and the other in $R$ and there is an edge between each vertex in $C$ and each vertex in $R$; and (iii) every vertex in $C$ has at least one neighbor in $S$ and every vertex in $S$ has at least one non-neighbor in $C$. Note that $R$ can be empty. We call a pseudo-split graph with vertex-set partition into sets~$S,C$ and $R$ an \emph{$(S,C,R)$-pseudo-split graph}.

\begin{figure}[htb]
    \centering
    \begin{tikzpicture}[scale=1]
        \tikzset{vertex/.style={draw, circle, minimum size=7pt,inner sep=1pt}};
        \node[vertex] (s1) at (-4,-2.5) {};
        \node[vertex,minimum size=11pt] (s21) at (-0.5,-2.5) {$v$};
        \node[vertex,minimum size=11pt] (s22) at (0.5,-2.5) {$u$};
        \node[vertex] (s3) at (4,-2.5) {};

        \node[vertex] (c1) at (-2,0) {};
        \node[vertex] (c2) at (0,0) {};
        \node[vertex] (c3) at (2,0) {};

        \node[vertex] (r1) at (-3,2.5) {};
        \node[vertex] (r2) at (-1,2.5) {};
        \node[vertex] (r3) at (1,2.5) {};
        \node[vertex] (r4) at (3,2.5) {};

        \path[-] (s1) edge (c2) edge (c3);
        \path[-] (s21) edge (c1) edge (c3);
        \path[-] (s22) edge (c1) edge (c3) edge (s21);
        \path[-] (s3) edge (c1) edge (c2);
        \draw[-] (c1) to (c2) to (c3) to[bend left=15] (c1);

        \draw[-] (r1) to (r2);
        \draw[-] (r3) to (r4);
        \path[-] (r1) edge (c1) edge (c2) edge (c3);
        \path[-] (r2) edge (c1) edge (c2) edge (c3);
        \path[-] (r3) edge (c1) edge (c2) edge (c3);
        \path[-] (r4) edge (c1) edge (c2) edge (c3);

        \draw[-,blue] (-4.5,-3) rectangle (4.5,-2);
        \node () at (5,-2.5) {\Large{$S$}};
        \draw[-,red] (-2.5,0.5) rectangle (2.5,-0.5);
        \node () at (3,0) {\Large{$C$}};
        \draw[-,green] (-3.5,3) rectangle (3.5,2);
        \node () at (4,2.5) {\Large{$R$}};
    \end{tikzpicture}
    \caption{\label{fig:quasispider}Example of a quasi-spider obtained from an $(S,C,R)$-thick spider by replacing a vertex from $S$ by a $K_2$ with the vertices $v$ and $u$.}
\end{figure}
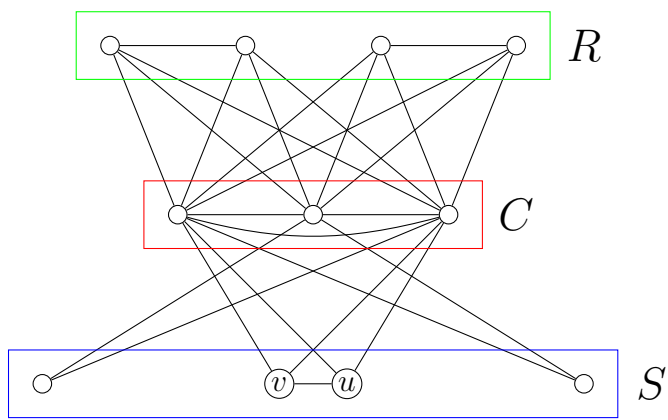
A \emph{spider} \cite{hoang1983} is a pseudo-split graph such that $|S|=|C|$, and there is a bijection $f : C \rightarrow S$ such that for every $v \in C$, either $N(v) = \{f(v)\}$, in which case $G$ is a \emph{thin spider}; or $N(v) = S \setminus \{f(v)\}$, in which case $G$ is a \emph{thick spider}. We refer to a (thick/thin) spider with vertex-set partitioned into sets $S,C$ and $R$ as an $(S,C,R)$-(thick/thin) spider.

A \emph{quasi-spider} is obtained from an $(S,C,R)$-spider by replacing one vertex in $S \cup C$ by a $K_2$ or $\overline{K_2}$, such that the new vertices have the same neighborhood as the original one, except for the mutual adjacency when obtaining a~$K_2$. \Cref{fig:quasispider} illustrates a quasi-spider obtained from an $(S,C,R)$-thick spider by replacing a vertex from $S$ by a $K_2$. 

Note that a quasi-spider obtained from an $(S,C,R)$-spider by either replacing a vertex from $S$ by a $\overline{K_2}$ or replacing a vertex from $C$ by a $K_2$ is not a spider, but it is a pseudo-split graph. Thus, the next lemma actually includes all quasi-spiders, since $G$ is either a pseudo-split graph or a quasi-spider obtained from an $(S,C,R)$-spider by either replacing a vertex from $S$ by a $K_2$ or replacing a vertex from $C$ by a $\overline{K_2}$.

\begin{lemma}
\label{lemma:convexityNumberpseudosplit}
Let either $G$ be an $(S,C,R)$-pseudo-split graph or $G'$ be an $(S,C,R)$-spider and $G$ be a quasi-spider obtained from $G'$ by either replacing~$v \in S$ by a $K_2$, or replacing $v \in C$ by a $\overline{K_2}$. Then
\begin{flalign*}
con(G) = \begin{cases}
|V(G)|-1&\text{if }\delta(G)=1\\
\alpha(G[R]) + |S|&\text{otherwise.}
\end{cases}
\end{flalign*}
\end{lemma}

\begin{proof}
If there is a vertex $v$ of $G$ such that $|N(v)| = 1$, then $v$ is in no cycle of $G$, which, by \Cref{lemma:conv1}, implies $con(G) = |V(G)|-1$. 

Henceforth, assume that $\delta(G) \geq 2$. If $G$ is a quasi-spider obtained from an $(S,C,R)$-spider $G'$, let $V(G) = C \cup R \cup S \cup \{v'\}$, where either $v \in S$ was replaced by a $K_2$ formed by the vertices $v$ and $v'$, or $v \in C$ was replaced by a $\overline{K_2}$ formed by the vertices $v$ and $v'$.

Let $Q$ be a largest independent set of $G[R]$. Since $T = Q \cup S$ is an independent set of $G$, whether $G$ is a pseudo-split graph or a quasi-spider graph, then $T$ is a convex set of $G$ different from $V(G)$. Hence, $con(G) \geq \alpha(G[R]) + |S|$.

Now, let $T \subseteq V(G)$ such that $|T| \geq \alpha(G[R]) + |S| + 1$ and denote $T_S = T \cap S, T_C = T \cap C$ and $T_R = T \cap R$. First, let $G$ be an $(S,C,R)$-pseudo-split graph. If $|T_R| > \alpha(G[R])$, $T$ is not an independent set of $G$. On the other hand, if $|T_R| \leq \alpha(G[R])$, either $T_C \geq 2$, or if $T_C = 1$, $T_S = S$, which also implies that $T$ is not an independent set of $G$.

For any two neighbors $v,w \in C \cup R$, $I(\{v,w\}) \supseteq C$, which implies that $I^2(\{v,w\}) = V(G)$, since every vertex in $S$ and~$R$ has at least two neighbors in~$C$ (remember that $\delta(G) \geq 2$). Furthermore, for similar reasons, for any two neighbors $v \in C$ and $w \in S$, $I^2(\{v,w\}) \supseteq C$, which implies that $I^3(\{v,w\}) = V(G)$. With that, we can infer that no convex set of $G$ different from $V(G)$ can contain two neighbors. Since $T$ is not an independent set, either $T$ is not convex or $T = V(G)$. We then conclude that $con(G) = \alpha(G[R]) + |S|$ if $G$ is an $(S,C,R)$-pseudo-split graph.

Now, let $G$ be a quasi-spider obtained from an $(S,C,R)$-spider $G'$ by replacing either $v \in S$ by a $K_2$, or $v \in C$ by a $\overline{K_2}$. Since $\delta(G) \geq 2$ and~$|S| \geq 2$, $G'$ is not a thin spider, which implies that $|S|=|C| \geq 3$. If $v \in S$, we can use the same reasoning we used in the case where $G$ is a pseudo-split graph to prove that $T$ is not an independent set, and for any two neighbors $u$ and~$w$, $H(\{u,w\}) = V(G)$. This allows us to reach the same conclusion that either~$T$ is not a convex set of $G$, or $T = V(G)$. The additional case that must be observed in order to conclude that $H(\{u,w\}) = V(G)$ is when $v \in S$, $u=v$ and $w = v'$. However, in this case, since $|C| \geq 3$ and $G'$ is a thick spider, $I^3(\{u,w\}) = V(G)$.

Finally, let $v \in C$ and denote $C' = C \cup \{v'\}$ and $T_{C'} = T \cap C'$. For any two neighbors $u$ and $w$, $H(\{u,w\}) = V(G)$, since $I(\{u,w\}) \supseteq C' \setminus \{v,v'\}$, and $I^2(\{u,w\}) \supseteq C'$. Thus, no convex set of $G$ different from $V(G)$ can contain two neighbors. Suppose by contradiction that $T$ is an independent set. Since $|T| \geq \alpha(G[R]) + |S| + 1$ and $T$ is an independent set, $1 \leq |T_{C'}| \leq 2$. If $|T_{C'}| = 2$, $T_{C'} = \{v,v'\}$ and $|T_S| \geq |S| - 1 \geq 2$. Thus, since $G'$ is a thick spider, there is at least one vertex in $T_S$ that is adjacent to $v$ and $v'$, a contradiction. If $|T_{C'}| = 1$, $|T_S| = |S|$, which again implies that there is some vertex in $T_S$ that is adjacent to the vertex in $T_{C'}$, a contradiction. Hence, we have that $T$ is not an independent set, which implies that either $T$ is not convex or $T = V(G)$.

Therefore, we conclude that $con(G) = \alpha(G[R]) + |S|$ if $G$ is a quasi-spider obtained from an $(S,C,R)$-spider by replacing either $v \in S$ by a $K_2$, or $v \in C$ by a $\overline{K_2}$.
\end{proof}

The \textsc{Independent Set} Problem is a well-known \NP-complete and \W[1]-complete problem when parameterized by $k$ \cite{cygan2015} that consists in, given a graph~$G$ and an integer $k$, determining whether $\alpha(G) \geq k$. The following complexity result follows from \Cref{lemma:convexityNumberpseudosplit}.
\begin{theorem}
\label{cor:convexityNumber}
\nameref{prob:convexityNumber} is \NP-complete and \W[1]-hard when parameterized by~$k$ for thick spiders.
\end{theorem}

\begin{proof}
This theorem follows directly from the fact that, given a graph $H$ and an integer $k$, $\alpha(H) = k$ if and only if $con(G) = \alpha(H) + |S| = k+3$, where $G$ is an $(S,C,R)$-thick spider such that $|S|=|C|=3$ and $G[R]$ is isomorphic to $H$. This equivalence, in turn, follows directly from \Cref{lemma:convexityNumberpseudosplit}.
\end{proof}

\begin{figure}[htb]
    \centering
    \begin{tikzpicture}[scale=1.2]
        \tikzset{class/.style={draw, rectangle, inner sep=3pt, font=\footnotesize}};
        \node[class] (exladen) at     (0,0){Extended $P_4$-laden};
        \node[class] (tidy) at        (-1.5,-1){$P_4$-tidy};
        \node[class] (laden) at       (1.5,-1){$P_4$-laden};        
        \node[class] (extend) at      (-4,-2){$P_4$-extendible};
        \node[class] (exsparse) at    (0,-2){Extended $P_4$-sparse};
        \node[class] (lite) at        (4,-2){$P_4$-lite};
        \node[class] (exreducible) at (-1.5,-3){Extended $P_4$-reducible};
        \node[class] (sparse) at      (1.5,-3){$P_4$-sparse};
        \node[class] (reducible) at   (0,-4){$P_4$-reducible};
        \node[class] (cograph) at     (0,-5){Cograph};

        \path[-] (exladen) edge (tidy) edge (laden);
        \path[-] (extend) edge (exreducible) edge (tidy);
        \path[-] (exsparse) edge (exreducible) edge (sparse) edge (tidy);
        \path[-] (lite) edge (sparse) edge (tidy) edge (laden);
        \path[-] (reducible) edge (exreducible) edge (sparse) edge (cograph);
    \end{tikzpicture}
    \caption{\label{fig:hierarquiaPoucosP4} Hierarchy of graphs with few $P_4$'s.}
\end{figure}
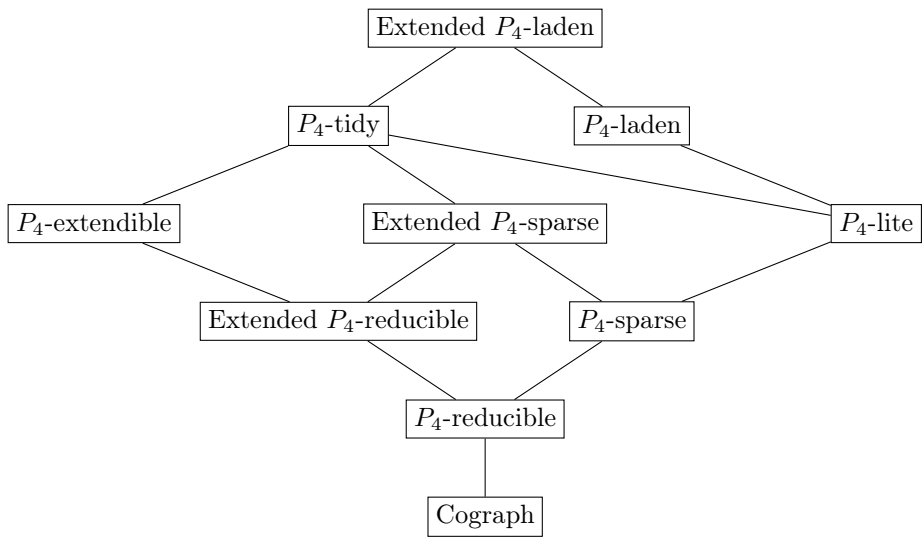

Created by Giakoumakis~\cite{giakoumakis1996}, extended $P_4$-laden graphs bear a close relationship to pseudo-split graphs. A graph is called \emph{extended $P_4$-laden} if every induced subgraph with at most six vertices that contains more than two induced $P_4$'s is a pseudo-split graph. The motivation to derive algorithms for extended $P_4$-laden graphs arises from the fact that it resides atop a widely studied hierarchy of classes containing graphs with few $P_4$’s, illustrated in \Cref{fig:hierarquiaPoucosP4}, including cographs, $P_4$-sparse, $P_4$-lite, $P_4$-laden and $P_4$-tidy graphs. In fact, the following theorem provides a natural way to derive efficient algorithms for this class by showing a convenient way to decompose an extended $P_4$-laden graph, which can be obtained in linear time.

\begin{theorem}[Giakoumakis \cite{giakoumakis1996}]
\label{thm:giakoumakis}
A graph $G$ is extended $P_4$-laden if and only if exactly one of the following conditions is satisfied:
\begin{itemize}
    \item $G$ is the union or join of two extended $P_4$-laden graphs;
    \item $G$ is an $(S,C,R)$-pseudo-split graph or a quasi-spider obtained from an $(S,C,R)$-spider, and, in both cases, $G[R]$ is an extended $P_4$-laden;
    \item $G$ is isomorphic to $C_5, P_5$ or $\overline{P_5}$;
    \item $G$ has at most one vertex.
\end{itemize}
\end{theorem}

\begin{theorem}
Let $G$ be an extended $P_4$-laden. Then, $con(G)$ can be computed in polynomial time.
\end{theorem}

\begin{proof}
By decomposing $G$ according to \Cref{thm:giakoumakis} in linear time, our algorithm could compute $con(G)$ recursively. However, the only case where it would need to recursively compute the convexity number of any of its extended $P_4$-laden subgraphs is the case where $G = G_1 \cup G_2$. Therefore, our algorithm simply calculates $con(C)$ for each $C \in \mathcal{C}(G)$, where $\mathcal{C}(G)$ is the set of all connected components of $G$, in the following manner:
\begin{itemize}
    \item if $C$ has at most one vertex, $con(C) = 0$;
    \item if $C$ is a $C_5$ or $\overline{P_5}$, $con(C) = 3$;
    \item if $C$ is a $P_5$, $con(C) = 4$;
    \item If $C$ is a join of two extended $P_4$-laden graphs, our algorithm computes $con(C)$ according to \Cref{lemma:convexityNumberunionjoin};
    \item If $C$ is either a pseudo-split graph or a quasi-spider, our algorithm computes $con(C)$ according to \Cref{lemma:convexityNumberpseudosplit}.
\end{itemize}
Then, our algorithm computes $con(G) = \max\limits_{C \in \mathcal{C}(G)} |V(G)| - |V(C)| + con(C)$ according to \Cref{lemma:convexityNumberunionjoin} applied to any number of unions.

Note that, depending on the case, in order to compute the convexity number of each connected component, our algorithm simply has to compute the size of given subgraphs, which can be done in linear time, and/or the size of a largest independent set of a given extended $P_4$-laden subgraph, which can be done in polynomial time \cite{giakoumakis1996,chvatal1987}.
\end{proof}

\section{Results on the Percolation Time}
\label{sec:percolationtime}
In this section, we present our results regarding the percolation time in the cycle convexity. We prove that the following problem is \NP-complete for any fixed~$k \geq 9$, but polynomial for fixed~$k \leq 2$ or when $G$ is a \emph{cactus}, that is, a connected graph in which each edge belongs to at most one unique cycle. Equivalently, a cactus is a connected graph where each two cycles share at most one vertex.

\problemClassic{Percolation Time}{A graph $G=(V,E)$ and an integer $k\geq0$.}{$pn(G) \geq k$?}{prob:percolationNumber}

All graphs in this section are connected, as the percolation time of a graph is the maximum percolation time of its connected components.
 
\subsection{Percolation Time is linear on cacti}

Let us begin by presenting a linear-time algorithm that computes $pn(G)$ when $G$ is a cactus. We can assume that $G$ has at least one cycle. We define the graph~$T_G = (\mathcal{C},A)$ where each vertex in $\mathcal{C}$ is a cycle of $G$ and we have an edge between two nodes $C_1$ and $C_2$ of $T_G$ if and only if $C_1$ and $C_2$ share a vertex in $G$. An \emph{articulation} of $G$ is a vertex whose removal increases the number of connected components of~$G$.
We say that a cycle of a cactus~$G$ is \emph{terminal} if it contains at most one articulation that belongs to at least two cycles of~$G$. Also, let $lp(H)$ be the longest induced path in the graph~$H$.

\begin{lemma}
\label{lemma:percolationnumbertg}
Given a cactus~$G=(V,E)$, $lp(T_G)$ can be computed in linear time in the size of $G$.
\end{lemma}

\begin{proof}
Let $\mathcal{C}$ be the set of cycles of $G$ and $T'_G = (\mathcal{C} \cup V, A)$ be the forest we get when $\mathcal{C}$ and $V$ are independent sets and we have an edge between $w \in V$ and $C \in \mathcal{C}$ if and only if $C$ contains $w$ in $G$. Any path in $T'_G$ alternates between a vertex in $V$ and a vertex in $\mathcal{C}$. Furthermore, since $G$ is a cactus, any maximal path in $T'_G$, which is a path that cannot be further extended, starts and ends at a vertex in $V$ and thus has an odd size. Also, the size of~$T'_G$ is linear in the size of $G$, since the number of cycles in a cactus is linear in its number of vertices.

By the definition of $T'_G$ and the fact that any two cycles in $G$ share at most one vertex, it follows that, if $v_0,C_0,v_1,C_1,\ldots,C_k,v_{k+1}$ is a path of $T'_G$, then $C_0,C_1,\ldots,C_k$ is an induced path in $T_G$ and, conversely, if $C_0,C_1,\ldots,C_k$ is an induced path in $T_G$, then there are vertices $v_0,v_1,v_2,\ldots,v_{k+1}$ such that $v_0,C_0,v_1,C_1,\ldots,C_k,v_{k+1}$ is a path of $T'_G$.

This implies that, $lp(T_G) = \frac{k-1}{2}$, where $k$ is the length of a maximum path of $T'_G$. Since there are a linear number of cycles in $G$, $T'_G$ can be constructed from $G$ in linear time. Thus, since we can compute the maximum path of $T'_G$, which is a forest, in linear time, we conclude that $lp(T_G)$ can be computed in linear time in the size of $G$.
\end{proof}

\begin{lemma}
\label{lemma:percolationnumbercacto}
Let~$G=(V,E)$ be a cactus with at least one cycle. Then ~$pn(G) = lp(T_G)+1$.
\end{lemma}

\begin{proof}
Let~$G=(V,E)$ be a cactus with at least one cycle and $k$ an integer. Let us prove that~$pn(G) \geq k+1$ if and only if there is an induced path of size $k$ in $T_G$.

Suppose that there exists a hull set $S$ of $G$ and a vertex $v_{k+1}$ such that~$v_{k+1} \in I^{k+1}(S) \setminus I^k(S)$. By the definition of a hull set in the cycle convexity, if $w \in I^x(S) \setminus I^{x-1}(S)$, there exists a cycle containing $w$ such that all other vertices of the cycle are in $I^{x-1}(S)$. Moreover, if $x > 1$, at least one vertex in the cycle must be in $I^{x-1}(S) \setminus I^{x-2}(S)$.

In fact, using this reasoning, given~$v_{k+1} \in I^{k+1}(S) \setminus I^k(S)$, we can define inductively the following sequence of cycles in $G$: for all $i = k, k-1, \ldots, 0$, let $C_i$ be a cycle that contains the vertex $v_{i+1}$ such that all other vertices of~$C_i$ are in $I^i(S)$, and, if~$i > 0$, there is a vertex $v_i$ in $I^i(S) \setminus I ^{i-1}(S)$. \Cref{fig:percolationNumbercactoida} illustrates an example of a cactus $G$ and vertices $v_1,v_2$ and $v_3$ based on a hull set $S$, with each cycle of the sequence $C_2,C_1,C_0$ highlighted. Note that any two consecutive cycles $C_j$ and $C_{j-1}$ of this sequence share the vertex $v_j$ and any two non-consecutive cycles of this sequence do not share vertices, since $G$ is a cactus. Thus, we get that this sequence of cycles also denotes an induced path of size $k$ in $T_G$.

\begin{figure}[htb]
\centering
\begin{tikzpicture}[scale=2]
    \tikzset{vertex/.style={draw, circle, minimum size=12pt,inner sep=1pt}};
    \tikzset{hullset/.style={draw, circle, minimum size=12pt,inner sep=2pt,fill=black!25}};
    \node[hullset] (v0) at (0,0){};
    \node[hullset] (c11) at (0.2,0.5){};
    \node[hullset] (c12) at (0.8,0.5){};
    \node[hullset] (c13) at (0.5,-0.5){};
    \node[vertex] (v1) at (1,0){\footnotesize{$v_1$}};
    \node[hullset] (c21) at (1.5,0.5){};
    \node[vertex] (c22) at (1.5,-0.5){};
    \node[hullset] (c22p1) at (1.2,-1){};
    \node[hullset] (c22p2) at (1.8,-1){};
    \node[vertex] (v2) at (2,0){\footnotesize{$v_2$}};
    \node[hullset] (v2p1) at (2,1){};
    \node[vertex] (v2p2) at (2.5,1){};
    \node[hullset] (c31) at (3,0.5){};
    \node[vertex] (v3) at (3,-0.5){\footnotesize{$v_3$}};

    \node () at (0.5,0) {\large{$C_0$}};
    \draw[-] (v0) to (c11) to (c12) to (v1) to (c13) to (v0);
    \node () at (1.5,0) {\large{$C_1$}};
    \draw[-] (v1) to (c21) to (v2) to (c22) to (v1);
    \node () at (2.65,0) {\large{$C_2$}};
    \draw[-] (v2) to (c31) to (v3) to (v2);

    \draw[-] (v2) to (v2p1) to (v2p2) to (v2);
    \draw[-] (c22) to (c22p1) to (c22p2) to (c22);
\end{tikzpicture}
\caption{An example of a hull set $S$ of a cactus, represented by the gray vertices. The vertices $v_1,v_2$ and $v_3$ and sequence of cycles $C_2,C_1,C_0$ are presented as defined.}
\label{fig:percolationNumbercactoida}
\end{figure}
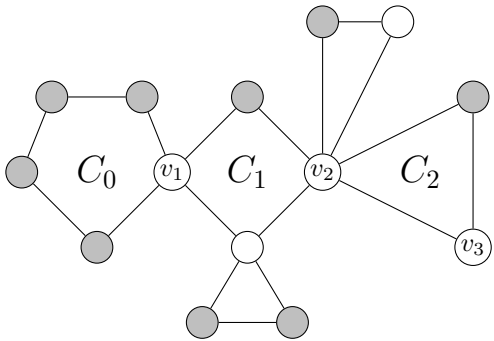

Now, let us prove that if there is an induced path of size $k$ in $T_G$, $pn(G) \geq k+1$. If $G$ is a cactus with at least one cycle, we define $L_G$ to be the set of vertices that belong to only one cycle and this cycle is the end of some maximum induced path of $T_G$, that is, this cycle is terminal. To demonstrate what we want, let us prove by induction on $k \geq 0$ that, for any cactus $G$ with at least one cycle such that $lp(T_G) = k$ and any $v \in L_G$, there is a hull set $S$ of $G$ such that $v \in I^{k+1}(S) \setminus I^k(S)$.

Let $G$ be a cactus such that $lp(T_G) = 0$ and $v \in L_G$. Since $lp(T_G) = 0$, $G$ has exactly one cycle, and thus $v$ is in this cycle. Hence, the set $S = V(G) \setminus \{v\}$ is a hull set such that $v \in I(S) \setminus S$.

Now, let $k \geq 1$ and suppose that for any cactus $G$ such that $lp(T_G) = k-1$ and $v \in L_G$, there exists a hull set $S$ of $G$ such that $v \in I^k(S) \setminus I^{k-1}(S)$. Let $G$ be a cactus such that $lp(T_G) = k$ and $v \in L_G$. Let $C_k$ be the unique cycle to which $v$ belongs, $C_0,C_1,\ldots,C_k$ be a maximum induced path of~$T_G$, and $v'$ be the only vertex common to cycles $C_k$ and $C_{k-1}$. Finally, let $G' = G - (L_G \setminus V(C_0))$. Note that neither the cycle $C_k$ nor the vertex $v$ is present in $G'$, but $v'$ is and $v' \in L_{G'}$.

By the definitions of $L_G$ and the fact that $C_0,C_1,\ldots,C_k$ is a maximum induced path of $T_G$, $lp(T_{G'}) = k-1$ and thus the path $C_0,C_1,\ldots,C_{k-1}$ is induced and maximum in $T_{G'}$. Therefore, by the inductive hypothesis, there is a hull set $S'$ of $G'$ such that $v' \in I^k(S') \setminus I^{k-1}(S')$ as $v' \in L_{G'}$. Let $S = S' \cup (L_G \setminus (V(C_0) \cup Q))$, where $Q$ is any set of vertices containing $v$ and exactly one vertex from each cycle of $G$ that is not in $G'$, except the cycle~$C_k$, that is not an articulation of $G$. \Cref{fig:percolationNumbercactovolta} illustrates an example with the sequence of cycles $C_0,C_1,\ldots,C_k$, for $k = 2$, the vertices $v$ and $v'$, and the sets $S', L_G$ and $Q$ represented by vertices colored black, outlined in red and colored gray respectively.

\begin{figure}[htb]
\centering
\begin{tikzpicture}[scale=2]
    \tikzset{vertex/.style={draw, circle, minimum size=12pt,inner sep=1pt}};
    \tikzset{hullsetSl/.style={draw, circle, minimum size=12pt,inner sep=2pt,fill=black!75}};
    \tikzset{hullsetLG/.style={draw=red, circle, minimum size=12pt,inner sep=2pt}};
    \tikzset{hullsetQ/.style={draw=red, circle, minimum size=12pt,inner sep=2pt, fill = black!25}};
    \node[hullsetSl,hullsetLG,thick] (v0) at (0,0){};
    \node[hullsetSl,hullsetLG,thick] (c11) at (0.2,0.5){};
    \node[hullsetSl,hullsetLG,thick] (c12) at (0.8,0.5){};
    \node[hullsetSl,hullsetLG,thick] (c13) at (0.5,-0.5){};
    \node[vertex] (v1) at (1,0){};
    \node[hullsetSl] (c21) at (1.5,0.5){};
    \node[hullsetSl] (c22) at (1.5,-0.5){};
    \node[hullsetLG] (c22p1) at (1.2,-1){};
    \node[hullsetQ] (c22p2) at (1.8,-1){};
    \node[vertex] (v2) at (2,0){\footnotesize{$v'$}};
    \node[hullsetLG] (v2p1) at (2,1){};
    \node[hullsetQ] (v2p2) at (2.5,1){};
    \node[hullsetLG] (c31) at (3,0.5){};
    \node[hullsetQ] (v3) at (3,-0.5){$v$};

    \node () at (0.5,0) {\large{$C_0$}};
    \draw[-] (v0) to (c11) to (c12) to (v1) to (c13) to (v0);
    \node () at (1.5,0) {\large{$C_1$}};
    \draw[-] (v1) to (c21) to (v2) to (c22) to (v1);
    \node () at (2.65,0) {\large{$C_2$}};
    \draw[-] (v2) to (c31) to (v3) to (v2);

    \draw[-] (v2) to (v2p1) to (v2p2) to (v2);
    \draw[-] (c22) to (c22p1) to (c22p2) to (c22);
\end{tikzpicture}
\caption{An example with the sequence of cycles $C_0,C_1,\ldots,C_k$, for $k = 2$, the vertices~$v$ and $v'$, and the sets $S', L_G$ and $Q$ represented by the black vertices, vertices outlined in red and gray vertices, respectively.}
\label{fig:percolationNumbercactovolta}
\end{figure}
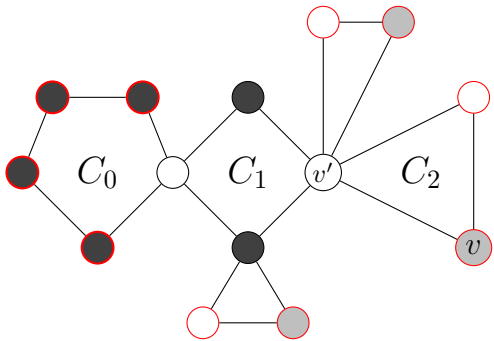

As $V(C_0) \subseteq V(G')$ and $V(G') \subseteq H(S') \subseteq H(S)$, to show that $S$ is a hull set of $G$, we just need to show that the vertices in $Q$ are in $H(S)$. Since, for each $w \in Q$, considering $C_w$ as the cycle that contains $w$ in $G$, we have that (i) there is a vertex in $C_w$ that is in $H(S')$ and, consequently, is in~$H(S)$, which is the only vertex that is also a vertex of $G'$; and (ii) all the other vertices of $C_w$, except $w$, are in the set $S$ itself, we can indeed conclude that~$Q \subseteq H(S)$ and, thus, $S$ is a hull set of $G$.

Since $S \setminus S'$ contains only vertices in terminal cycles and it neither contains the articulation of such cycles nor all vertices of any terminal cycle, for any~$w \in V(G')$, we have that $w \in I^k(S')$ if and only if $w \in I^k(S)$. Hence, since $v' \in I^k(S') \setminus I^{k-1}(S')$, this implies that $v' \in I^k(S) \setminus I^{k-1}(S)$ and then we can conclude that $v \in I^{k+1}(S) \setminus I^k(S)$.
\end{proof}

\begin{theorem} 
Given a cactus~$G$, $pn(G)$ can be computed in linear time.
\end{theorem}

\begin{proof}
The theorem follows directly from \Cref{lemma:percolationnumbertg,lemma:percolationnumbercacto}, and the fact that~$pn(G) = 0$ if and only if $G$ has no cycles.
\end{proof}

\subsection{Percolation Time is polynomial for fixed \texorpdfstring{$k \leq 2$}{k ≤ 2}}

Now we present a cubic algorithm to solve the \nameref{prob:percolationNumber} problem when $k \leq 2$. For any graph $G$, $pn(G) \geq 0$, and~$pn(G) \geq 1$ if and only if~$G$ has at least one cycle. Thus, for fixed~$k \leq 1$, the problem can be solved is linear time. Henceforth, in this subsection, we focus on the case where~$k=2$.

Let~$G=(V,E)$ be a graph and~$v,w \in V$ be two adjacent vertices. We define~$R_G(v,w)$ as the set of all subgraphs of~$G$ induced by a set of vertices $A \cup B$, such that $G[A]$ is a chordless cycle $C$ containing the edge~$vw$, $G[B]$ is a (not necessarily induced) cycle $C'$ that contains~$w$ and does not contain~$v$, and~$v$ has at most one neighbor in the set~$B \setminus \{w\}$. Note that the cycles~$C$ and~$C'$ may have other vertices aside from~$w$ in common.

\begin{lemma}
\label{lem:percolationNumberk2Estrutura}
Let~$G=(V,E)$ be a graph. There are neighbors~$v,w \in V$ such that~$R_G(v,w) \neq \varnothing$ if and only if there is a set~$S \subseteq V$ with $I^2(S) \setminus I(S) \neq \varnothing$. Furthermore, given two neighbors~$v,w \in V$ in which~$R_G(v,w) \neq \varnothing$, such a set~$S$ can be computed in linear time in the size of $G$.
\end{lemma}

\begin{proof}
Let~$v,w \in V(G)$ be two neighbors. Moreover, let~$G[R] \in R_G(v,w)$ be an induced subgraph of~$G$. We consider a vertex subset~$R = V(C) \cup V(C')$, where~$C$ is an induced cycle of~$G$ that contains the edge~$vw$ and~$C'$ is a cycle that contains~$w$ and does not contain~$v$ such that~$v$ has at most one neighbor in the set~$V(C') \setminus \{w\}$.

If~$v$ has only one neighbor $u$ in~$V(C') \setminus \{w\}$, we have that the set~$S = V(C') \setminus \{w\} $ is such that~$v \in I^2(S) \setminus I(S)$. This happens because (i)~$v$ has only one neighbor in~$S$ and thus $v \notin I(S)$, and (ii) $v \in I^2(S)$, since $w \in I(S)$, there is a path from $u$ to $w$ through $C'$, and both vertices are adjacent to $v$. An example of this case is illustrated in \Cref{fig:percolationNumberk2EstruturaIdaCaso1}.

\begin{figure}[htb]
\centering
\begin{tikzpicture}[scale=1]
    \tikzset{vertex/.style={draw, circle, minimum size=10pt,inner sep=1pt, font=\footnotesize}};
    \tikzset{hullset/.style={draw, circle, minimum size=10pt,inner sep=1pt,fill=black!40, font=\footnotesize}};
    \node[hullset, minimum size=7pt] (q1) at (4,0){};
    \node[hullset] (u) at (4,2){$u$};
    \node[vertex] (v) at (2,2){$v$};
    \node[vertex] (w) at (2,0){$w$};
    \node[vertex, minimum size=7pt] (c) at (0,0){};
    \draw[-,blue] (v) to (c) to (w) to (v);
    \draw[-] (v) to (u) to (c);
    \draw[-,red] (u) to (w) to (q1) to (u);
        
\end{tikzpicture}
\caption{An example where~$v$ has exactly one neighbor in~$V(C') \setminus \{w\}$. The cycle~$C$ is denoted by the blue edges, the cycle~$C'$ by the red ones and~$S$ by the gray vertices.}
\label{fig:percolationNumberk2EstruturaIdaCaso1}
\end{figure}

If~$v$ has no neighbors in the set~$V(C') \setminus \{w\}$, let~$S = R \setminus \{v,w\}$. Since, $V(C') \setminus \{w\} \subseteq S$,~$w \in I(S)$, which implies that~$v \in I^2(S)$. On the other hand, since $v$ has only one neighbor in $S$, which is its neighbor in $C$ other than $w$, $v \notin I(S)$. An example of this case is illustrated in \Cref{fig:percolationNumberk2EstruturaIdaCaso2}.

\begin{figure}[htb]
\centering
\begin{tikzpicture}[scale=1]
    \tikzset{vertex/.style={draw, circle, minimum size=10pt,inner sep=1pt, font=\footnotesize}};
    \tikzset{hullset/.style={draw, circle, minimum size=7pt,inner sep=2pt,fill=black!40, font=\footnotesize}};
    \node[hullset] (q1) at (0,0){};
    \node[hullset] (q2) at (-0.5,0.5){};
    \node[vertex] (v) at (0:1.5){$v$};
    \node[vertex] (w) at (72:1.5){$w$};
    \node[hullset] (c1) at (144:1.5){};
    \node[hullset] (c2) at (216:1.5){};
    \node[hullset] (c3) at (288:1.5){};
    
    \draw[-,blue] (v) to (w) to (c1) to (c2) to (c3) to (v);
    \draw[-,red] (q2) to (w) to (q1) to (c2) to (q2);
    \draw[-] (q2) to (q1);
        
\end{tikzpicture}
\caption{An example of the case where~$v$ has no neighbors in~$V(C') \setminus \{w\}$. The cycle~$C$ is denoted by the blue edges, the cycle~$C'$ by the red ones and~$S$ by the gray vertices.}
\label{fig:percolationNumberk2EstruturaIdaCaso2}
\end{figure}
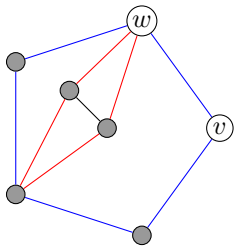

Suppose there are a vertex~$v' \in V(G)$ and a set~$S \subseteq V(G)$ such that~$v' \in I^2(S)\setminus I(S)$. Thus, there is a cycle containing $v'$. Let~$C$ be a smallest cycle of~$G$ that contains~$v'$ such that, for all~$q \in V(C) \setminus \{v'\}$,~$q \in I (S)$, and there exists~$w' \in V(C) \setminus \{v'\}$ such that~$w' \in I(S) \setminus S$. Let~$D$ be a cycle that contains~$w'$ such that, for all~$q \in V(D) \setminus \{w'\}$,~$q \in S$. Since~$v' \in I^2(S)\setminus I(S)$, $v'$ has at most one neighbor in~$V(D) \setminus \{w'\}$. Note that the cycles~$C$ and~$D$ do exist, since~$v' \in I^2(S)\setminus I(S)$ and~$w' \in I(S) \setminus S$. Also, since $C$ is the smallest cycle of~$G$ that contains~$v'$ such that, for all~$q \in V(C) \setminus \{v'\}$,~$q \in I (S)$, and there is no cycle $Q$ that contains~$v'$ such that, for all~$q \in V(Q) \setminus \{v'\}$,~$q \in S$, then we can infer that $C$ is an induced cycle.

We can see the induced cycle~$C$ as two vertex-disjoint paths both with ends at~$v'$ and~$w'$, one of which we name~$P$. Define~$v$ as the first vertex that is not in~$D$ and has at most one neighbor in~$V(D) \setminus \{w'\}$ we find when we go through $P$ starting at $w'$, and let $w$ be its neighbor in $P$ that is closest to~$w'$. Since~$v'$ is not in~$D$ and has at most one neighbor in~$V(D) \setminus \{w'\}$, such pair of vertices $v$ and $w$ are well defined.

First, suppose that $w \in V(D)$. If, $w' = w$, $v$ has at most one neighbor in~$V(D) \setminus \{w\}$, as $v$ has at most one neighbor in~$V(D) \setminus \{w'\}$. If $w' \neq w$,~$v$ has exactly one neighbor in~$V(D) \setminus \{w'\}$, which is $w$. Hence, $v$ has at most one neighbor in~$V(D) \setminus \{w\}$. Thus, either way, $v$ has at most one neighbor in~$V(D) \setminus \{w\}$ and $v \notin V(D)$, which implies that~$G[V(C) \cup V(D)] \in R(v,w)$. \Cref{fig:percolationNumberk2EstruturaVoltaCaso1} illustrates this case.

\begin{figure}[htb]
\centering
\begin{tikzpicture}[scale=1]
    \tikzset{vertex/.style={draw, circle, minimum size=15pt,inner sep=1pt, font=\footnotesize}};
    \tikzset{hullset/.style={draw, circle, minimum size=15pt,inner sep=1pt,fill=black!40, font=\footnotesize}};
    \node[hullset, minimum size=7pt] (q1) at (0,0){};
    \node[hullset, minimum size=7pt] (q2) at (-0.5,0.5){};
    \node[vertex] (vl) at (0:1.5){$v'$};
    \node[vertex] (wl) at (72:1.5){$w'$};
    \node[hullset, minimum size=7pt] (c1) at (144:1.5){};
    \node[hullset] (w) at (216:1.5){$w$};
    \node[hullset] (v) at (288:1.5){$v$};
    
    \draw[-,blue] (vl) to (wl) to (c1) to (w) to (v) to (vl);
    \draw[-,red] (q2) to (wl) to (q1) to (w) to (q2);
    \draw[-] (q2) to (c1);
        
\end{tikzpicture}
\caption{An example for the case where $w \in V(D)$ and $w \neq w'$. Set $S$ is depicted by the gray vertices. The cycle~$C$ is denoted by the blue edges and the cycle~$D$ by the red edges.}
\label{fig:percolationNumberk2EstruturaVoltaCaso1}
\end{figure}
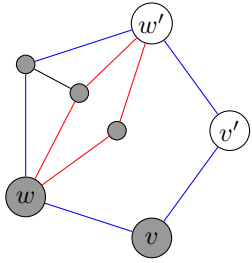

Now, suppose that $w \notin V(D)$. Then, by its own definition, $w$ has two or more neighbors in~$V(D) \setminus \{w'\}$ because, otherwise, $w$ would be actually~$v$ and not $w$. Let $x$ and $y$ be such neighbors of $w$ in~$V(D) \setminus \{w'\}$. Since $D$ is a (not necessarily chordless) cycle, there exists a path $P$ with ends in $x$ and $y$ in $D$ that avoids $w'$ entirely. Let $C'$ be the cycle formed by $P$ and the vertex $w$. Note that $V(P) \subseteq V(D)$, which implies that $V(C') \setminus \{w\} \subseteq V(D)$. So, since $v$ has at most one neighbor in~$V(D) \setminus w'$ and $w' \notin V(C')$, then~$v$ also has at most one neighbor in~$V(C') \setminus w$. We then conclude that $v$ has at most one neighbor in~$V(C') \setminus \{w\}$ and $v \notin V(D)$, which implies that~$G[V(C) \cup V(C')] \in R(v,w)$. \Cref{fig:percolationNumberk2EstruturaVoltaCaso2} illustrates this case.

\begin{figure}[htb]
\centering
\begin{tikzpicture}[scale=1.2]
    \tikzset{vertex/.style={draw, circle, minimum size=15pt,inner sep=1pt, font=\footnotesize}};
    \tikzset{hullset/.style={draw, circle, minimum size=15pt,inner sep=1pt,fill=black!40, font=\footnotesize}};
    \node[hullset] (x) at (108:-0.5){$x$};
    \node[hullset] (y) at (108:0.5){$y$};
    \node[hullset, minimum size=7pt] (c1) at (0:1.5){};
    \node[hullset, minimum size=7pt] (c2) at (30:0.7){};
    \node[vertex] (wl) at (72:1.5){$w'$};
    \node[hullset] (w) at (144:1.5){$w$};
    \node[hullset] (v) at (216:1.5){$v$};
    \node[vertex] (vl) at (288:1.5){$v'$};
    
    \draw[-,blue] (c1) to (wl) to (w) to (v) to (vl) to (c1);
    \draw[-,red] (wl) to (y);
    \draw[dashed,red] (y) to (x);
    \draw[-,red] (x) to (c2) to (wl);
    \draw[dashed] (x) to (w) to (y);
    \draw[-] (v) to (x);
        
\end{tikzpicture}
\caption{An example for the case where $w \notin V(D)$. Set $S$ contains the vertices colored as gray. The cycle~$C$ is denoted by the blue edges, the cycle~$D$ by the red edges, and the cycle $C'$ by the dashed edges.}
\label{fig:percolationNumberk2EstruturaVoltaCaso2}
\end{figure}
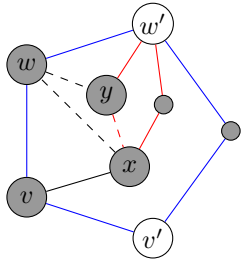

Given neighbors $v,w \in V$, in order to determine whether~$R(v,w) \neq \varnothing$ or not and, if in fact~$R(v,w)  \neq \varnothing$, to construct a subgraph in set~$R(v,w)$, we just need to construct, if it exists, an induced cycle~$C$ that contains the edge~$vw$ and then construct, if it exists, a cycle~$C'$ that contains~$w$ and does not contain~$v$ such that~$v$ has at most one neighbor in the set~$V(C') \setminus \{w\}$. If there are such~$C$ and~$C'$,~$G[V(C) \cup V(C')] \in R(v,w)$, otherwise,~$R(v ,w) = \varnothing$. Checking the existence of $C$ and $C'$ can be done independently and in linear time both with simple modifications of the standard DFS algorithm.
\end{proof}

\begin{lemma}
\label{lem:percolationNumberk2Extensao}
Let~$G=(V,E)$ be a graph,~$v \in V$ and~$Q \subseteq V$. If~$v \in I^2(Q)\setminus I(Q)$, then there exists a hull set~$S \supseteq Q$ of~$G$ such that~$v \in I^2(S) \setminus I(S)$. Additionally, such hull set~$S$ can be computed in cubic time.
\end{lemma}

\begin{proof}
Suppose that~$v \in I^2(Q)\setminus I(Q)$. Let~$R_0 = \{\}$. Also, let~$V' = \{v\} \cup Q \cup (V \setminus H(Q))$, that is,~$V'$ is the set of vertices consisting in the vertex~$v$ and all the vertices that are not in $H(Q) \setminus Q$. Define, for~$i \geq 1$,~$w_i$ to be a vertex in~$V \setminus H(Q)$ that belongs to some cycle~$C$ of~$G[V' \setminus R_{i-1}]$ which also contains~$v$ if such a cycle exists. If such a cycle~$C$ exists, let~$R_i = R_{i-1} \cup \{w_i\}$. If, for a given value of~$i=k+1$, such a cycle~$C$ does not exist, the sequence of sets~$R_i$ ends at the index~$k$ and we define~$S = V' \setminus (\{v\} \cup R_k)$. The $\Cref{fig:percolationNumberExampleSet}$ illustrates an example with the vertex $v$ and the sets~$Q, H(Q),R_k$ and~$S$.

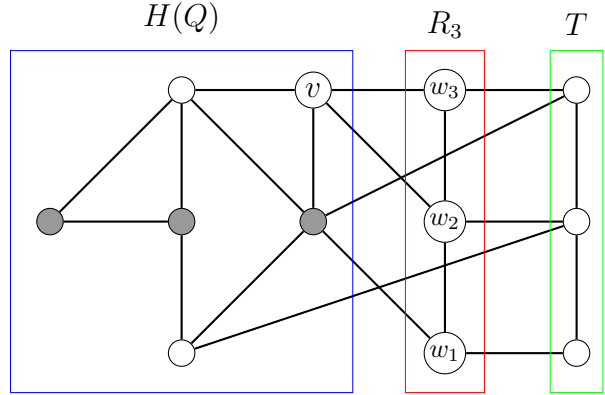
\begin{figure}[htb]
\centering
\begin{tikzpicture}[scale=1.3]
    \tikzset{vertex/.style={draw, circle, minimum size=10pt,inner sep=1pt}};
    \tikzset{hullset/.style={draw, circle, minimum size=10pt,inner sep=2pt,fill=black!40}};
    \node[hullset] (q1) at (1.25,0){};
    \node[hullset] (q2) at (0,0){};
    \node[hullset] (q3) at (-1.25,0){};
    \node[vertex,inner sep=2pt] (v) at (1.25,1){$v$};
    \node[vertex] (hq1) at (0,1){};
    \node[vertex] (hq2) at (0,-1){};

    \node[vertex] (w1) at (2.5,-1){\footnotesize{$w_1$}};
    \node[vertex] (w2) at (2.5,0){\footnotesize{$w_2$}};
    \node[vertex] (w3) at (2.5,1){\footnotesize{$w_3$}};
    \node[vertex] (ot1) at (3.75,1){};
    \node[vertex] (ot2) at (3.75,0){};
    \node[vertex] (ot3) at (3.75,-1){};
        
    \draw[-,thick] (q1) to (hq2) to (q2) to (q3) to (hq1) to (v) to (w2) to (w3) to (v) to (q1) to (hq1) to (q2);

    \draw[-,thick] (w1) to (q1) to (ot1) to (w3);
    \draw[-,thick] (ot1) to (ot2) to (w2) to (w1) to (ot3) to (ot2) to (hq2);

    \node () at (0,1.55){$H(Q)$};
    \node () at (2.5,1.5){$R_3$};
    \node () at (3.75,1.5){$T$};
    \draw[blue,ultra thin] (-1.55,1.3) rectangle (1.55,-1.3);
    \draw[red,ultra thin] (2.2,1.3) rectangle (2.8,-1.3);
    \draw[green,ultra thin] (3.55,1.3) rectangle (3.95,-1.3);
\end{tikzpicture}
\caption{An example of a graph with the vertex $v$ and the sets~$Q$, represented by the gray vertices, $H(Q), R_3$ and $S = Q \cup T$.}
\label{fig:percolationNumberExampleSet}
\end{figure}

Note that $Q \subseteq S$, which implies that $H(Q) \subseteq H(S)$. This also implies that either~$v \in I(S)$ or~$v \in I^2(S) \setminus I(S)$, since~$v \notin S$ and~$v \in I^2(Q) \setminus I(Q)$. Therefore, in order to prove what we want, we just need to show that~$S$ is a hull set and~$v \notin I(S)$.

Thus, first, let us prove that~$S$ is a hull set. To do this, we have to show that~$V \setminus H(Q) \subseteq H(S)$, since~$H(Q) \subseteq H(S)$. Let~$k$ be the last index of the sequence of sets~$R_i$ and let~$R_k=\{w_1,w_2,\ldots,w_k\}$. Each vertex in~$V \setminus H(Q)$ is either in~$R_k$ or in~$S$, so, in fact, we just need to show that~$R_k \subseteq H(S)$. First, remember that~$V' = \{v\} \cup Q \cup (V \setminus H(Q))$, which is the set consisting of basically every vertex that is not in $H(Q)$, except for $v$ and the vertices in $Q$. Let $T = V \setminus (H(Q) \cup R_k)$ and thus $V' \setminus R_k = \{v\} \cup Q \cup T$. Since $T \cup Q \subseteq S$ and $v \in H(S)$, then $V' \setminus R_k \subseteq H(S)$.

Now, let us show by induction on $i=k,k-1,\ldots,0$ that $V' \setminus R_i \subseteq H(S)$, which implies that $R_k \subseteq H(S)$, since $R_k \subseteq V' \setminus R_0$. We already proved that~$V' \setminus R_k \subseteq H(S)$. Let $0 \leq i \leq k-1$ and assume that $V' \setminus R_{i+1} \subseteq H(S)$. Since~$w_{i+1}$ is in a cycle in the subgraph~$G[V' \setminus R_i]$ and $V' \setminus R_{i+1} \subseteq H(S)$, we have $w_{i+1} \in H(S)$, which implies that $V' \setminus R_i \subseteq H(S)$. We then conclude that $S$ is a hull set.

Additionally, notice that~$v$ does not belong to any cycle composed only of elements of~$S$, as~$v$ does not belong to any cycle composed only of the vertices of the subgraph~$G[V' \setminus R_k]$, since~$R_k$ is the last one of its sequence. Therefore, we conclude that $v \notin I(S)$, as we wanted.

Given $G,v$ and $Q$, we can compute $H(Q)$ in cubic time, by computing, for each $k > 0$ and $q \notin I^{k-1}(Q)$, whether $q \in I^k(Q)$ in linear time with a~DFS. Finally, we can compute $R_k$, and consequently $S$, with an additional quadratic time by initially doing $R = \varnothing$ and successively computing whether there is a cycle in the subgraph $G[V' \setminus R]$ that has a vertex $v \in V \setminus H(Q)$ with a DFS. If there is such a cycle, we add $v$ to~$R$. Otherwise, the algorithm stops and returns $S = Q \cup (V \setminus (H(Q) \cup R))$. Each iteration can be done in linear time due to the DFS, and we have $O(|V(G)|)$ such iterations. This procedure computes $S$ from~$v, Q$, and $G$ in cubic time.
\end{proof}

\begin{theorem}
\label{thm:percolationtimek2}
It is possible to determine a hull set $S$ of a graph $G$ such that $I(S) \neq V(G)$, or that such hull set does not exist in cubic time.
\end{theorem}

\begin{proof}
Let $G = (V, E)$ be the input graph. The algorithm that computes a hull set $S$ of a graph $G$ such that $I(S) \neq V(G)$ if such set exists, consists initially of checking whether there is a pair of vertices $v $ and $w$ such that $R_G(v,w) \neq \varnothing$ in cubic time. If there is no such pair of vertices, the algorithm stops indicating that there is no hull set $S$ of $G$ such that~$I(S) \neq V(G)$. Otherwise, the algorithm computes in linear time a set $S$ such that~$I^2(S) \setminus I(S) \neq \varnothing$, and, finally, in cubic time, extends this set $S$ into a hull set $S' \supseteq S$ such that $I(S') \neq V(G)$. The correctness and complexity of the algorithm follow directly from \Cref{lem:percolationNumberk2Estrutura,lem:percolationNumberk2Extensao}.
\end{proof}

The following corollary is a direct consequence of \Cref{thm:percolationtimek2}.
\begin{corollary}
The problem \nameref{prob:percolationNumber} can be decided in cubic time for $k=2$.
\end{corollary}

\subsection{Percolation Time is \texorpdfstring{\NP}{NP}-complete for fixed \texorpdfstring{$k \geq 9$}{k ≥ 9}}

In this subsection, we show that the problem \nameref{prob:percolationNumber} is \NP-complete even for fixed $k \geq 9$. First, let us prove the following about the gadget in \Cref{fig:percolationNumberk9pendant}, which we call a \emph{perpetuation gadget}.

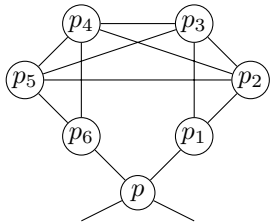
\begin{figure}[htb]
\centering
\begin{tikzpicture}[scale=1.5,auto]
    \tikzset{vertex/.style={draw, circle, minimum size=13pt,inner sep=1pt, font=\footnotesize}};
    \def \angle {51.43}
    \def \radius {1}
    \node[vertex] (q) at (0,0) {$p$};
    \node[vertex] (u1) at (0.5,0.5) {$p_1$};
    \node[vertex] (u2) at (-0.5,0.5) {$p_6$};
    \node[vertex] (u3) at (1,1) {$p_2$};
    \node[vertex] (u4) at (-1,1) {$p_5$};
    \node[vertex] (u5) at (0.5,1.5) {$p_3$};
    \node[vertex] (u6) at (-0.5,1.5) {$p_4$};

    \draw[-] (u1) to (u3) to (u5) to (u6) to (u4) to (u2) to (q) to (u1) to (u5) to (u4) to (u3) to (u6) to (u2);
    \draw[-] (-0.5,-0.25) to (q) to (0.5,-0.25);
    
\end{tikzpicture}
\caption{\emph{Perpetuation gadget}.}
\label{fig:percolationNumberk9pendant}
\end{figure}

\begin{lemma}
\label{lemma:percolationNumberk9pendant}
Let $G$ be a graph containing an induced subgraph $H$ isomorphic to the graph in \Cref{fig:percolationNumberk9pendant}, such that all vertices of $H$, except possibly for $p$, have no neighbors outside $H$. For any hull set $S$ of $G$, $V(H) \subseteq I^4(S)$.
\end{lemma}

\begin{proof}
Without loss of generality, assume that $S$ is a minimal hull set of $G$, that is, any subset of $S$ is not a hull set of $G$, and let~$S' = S \cap (V(H) \setminus \{p\})$. Let ~$A = \{p_1,p_2,p_3\}$,~$B = \{p_4,p_5,p_6\}$, and~$C = \{p_2,p_3,p_4,p_5\}$, where the vertices~$p_1,p_2,p_3,p_4,p_5$, and $p_6$ are as indicated in \Cref{fig:percolationNumberk9pendant}.

We can infer that~$|S'| > 1$ because, otherwise, since $p$ is an articulation of~$G$, we could never have $V(H) \subseteq H(S)$. Furthermore, for any $u \neq v$ in~$V(H)$, if either~$u,v \in A$ or~$u,v \in B$, then~$H(\{u ,v\}) = V(H)$. So, if~$|S'| \geq 3$, $S$ would not be minimal, since, by the pigeonhole principle, if~$|S'| \geq 3$, or~$|S \cap A| \geq 2$ or~$|S \cap B| \geq 2$. Hence, we can either choose any two vertices $v,w \in S' \cap A$ or any two vertices $v,w \in S' \cap B$ so the set $(S \setminus S') \cup \{v,w\}$ is a hull set.

As a result,~$S' = \{u,v\}$. This implies that or~$u,v \in A$,~or $u,v \in B$,~or even~$u,v \in C$, because only then~$V(H) \subseteq H(S)$, even if $p \in S$. Thus, regardless of which of the three sets $u$ and $v$ are, $V(H) \subseteq I^4(S')$, which implies that~$V(H) \subseteq I ^4(S)$.
\end{proof}

Given that the vertex $p$ and two of its neighbors $u$ and $v$ form a $K_3$ in~$G$, the purpose of a perpetuation gadget is simply to ensure that, for any $k \geq 4$, if $v \in I^k(S) \setminus I^{k-1}(S)$, then $u \in I^{k+1}(S)$. Now, we can prove the main result of this subsection.

\begin{theorem}
\nameref{prob:percolationNumber} is \NP-complete for any fixed~$k \geq 9$.
\end{theorem}

\begin{proof}
The problem is in \NP, since we can verify whether $S$ is a hull set of~$G$ such that $I^{k-1}(S) \neq V(G)$ in polynomial time.

In order to prove its \NP-hardness for $k=9$, we present a polynomial reduction from \kSAT{3}. In the end, we show how to generalize the reduction for any fixed $k > 9$. This reduction uses the same key ideas as the ones used in~\cite{benevides2015}, where the authors show that \nameref{prob:percolationNumber} is \NP-complete for any fixed $k \geq 4$ in the $P_3$ convexity.

Let $\varphi$ be an instance of \kSAT{3}, where the set of its clauses is denoted by $\mathcal{C} = \{C_1,C_2,\ldots,C_m\}$. Here, $C_i = \{\ell^i_1,\ell^i_2, \ell^i_3\}$ represents the literals of the $i$-th clause of $\varphi$, and the set of its variables is denoted by $\mathcal{X} = \{X_1,X_2, \ldots,X_n\}$. We proceed by constructing a graph $G$ from $\varphi$.

\begin{itemize}
    \item For each clause~$C_i$, we construct the gadget in \Cref{fig:percolationNumberk9}, where the loops at the vertices~$q_1^i,q_2^i,q_3^i,r_1^i,r_2^i,r_3^i$ and $z^i$ represent the perpetuation gadget, with each of these vertices at the same position as the vertex~$p$. Let~$V_i$ be the set of vertices in the clause gadget~$C_i$;
    \item For each pair of opposite literals $\ell^i_a$ and $\ell^j_b$ of the same variable, add the vertices~$y_{i,a,j,b}$ and~$y'_{i,a ,j,b}$. Form a~$K_3$ with the vertices~$w^i_a,w^j_b$ and~$y_{i,a,j,b}$ and add a perpetuation gadget with~$y'_{i,a,j,b}$ at the same position as the vertex~$p$. Let~$Y$ be the set of vertices~$y_{i,a,j,b}$ and~$Y'$ be the set of vertices~$y'_{i,a,j,b}$ added in this step;
    \item Add a vertex $x$ and, for each pair of vertices $y_{i,a,j,b}$ and $y'_{i,a,j,b}$, form a~$K_3$ with the vertices $y_{i,a,j,b},y'_{i,a,j,b}$ and $x$.
\end{itemize}

\begin{figure}[htb]
\centering
\begin{tikzpicture}[scale=0.85,auto]
    \tikzset{vertex/.style={draw, circle, minimum size=10pt,inner sep=1pt,font=\footnotesize}};
    \def \diff {17}
    \def \dist {1.25}
    \node[vertex] (z) at (0,0) {$z$};
    \draw[-] (z) edge[loop,in=270,out=330,looseness=6] (z);
    \foreach \i [evaluate=\i as \angle using ((\i-1)*120)] in {1,2,3}{
        \node[vertex,fill=black!25] (u\i) at (\angle:\dist) {$u_{\i}$};
        \node[vertex,fill=black!25] (v\i) at (\angle - \diff:2*\dist) {$v_{\i}$};
        \node[vertex] (w\i)  at (\angle:3*\dist) {$w_{\i}$};
        \node[vertex] (q\i) at (\angle-2.3*\diff:1.4*\dist) {$q_{\i}$};
        \draw[-] (u\i) to (w\i) to (v\i) to (u\i) to (q\i) to (v\i) to (u\i) to (z);
        \draw[-] (q\i) edge[loop,in=\angle-100,out=\angle-170,looseness=7] (q\i);
    }
    \foreach \i/\im [evaluate=\i as \angle using ((\i-1)*120 + 60)] in {1/2,2/3,3/1}{
        \node[vertex] (r\i) at (\angle:3*\dist-0.5) {$r_{\i}$};
        \draw[-] (u\i) to (u\im);
        \draw[-] (w\i) to (r\i) to (w\im);
        \draw[-] (w\i) to[bend right=17] (w\im);
        \draw[-] (r\i) edge[loop,in=\angle+45,out=\angle-45,looseness=5] (r\i);
    }
\end{tikzpicture}
\caption{The gadget for $C_i$. We omitted the superscript from each vertex's label. Each loop in a vertex $x$ represents a perpetuation gadget with $x$ being in the same position as the vertex $p$. The co-convex set $U_i$ is represented by the gray vertices.}
\label{fig:percolationNumberk9}
\end{figure}

Since each step of the construction has polynomial complexity, the construction has polynomial complexity. Let us prove that $\varphi$ is satisfiable if and only if there exists a hull set $S$ of $G$ such that $I^8(S) \neq V(G)$.

Suppose $\varphi$ is satisfiable. Given an assignment that satisfies $\varphi$, let $S$ be the subset of vertices of $G$ containing all vertices $p_1$ and $p_2$ of each perpetuation gadget in $G$ and, for each $1 \leq i \leq m$, it contains any vertex $u^i_a$ such that $\ell^i_a$ is evaluated to true, but exactly one for each $i$. By that definition of $S$, for each $1 \leq i \leq m$, $\{q^i_1,q^i_2,q^i_3,r^i_1,r^i_2,r^i_3,z^i\} \subseteq I^4(S) \setminus I^3(S)$, and $Y' \subseteq I^4(S)$. Thus, $\{u^i_1,u^i_2,u^i_3\} \subseteq I^5(S)$, $\{v^i_1,v^i_2,v^i_3\} \subseteq I^6(S)$, and, consequently, $\{w^i_1,w^i_2,w^i_3\} \subseteq I^7(S)$. Therefore, we conclude that, for all $1 \leq i \leq m$, $V_i \subseteq I^7(S)$.

Note that, for any $1 \leq a \leq 3$, either $w^i_a \in I^7(S) \setminus I^6(S)$ if $u^i_a \notin S$, or $w^i_a \in I^6(S) \setminus I^5(S)$ if $u^i_a \in S$. Since for any two opposite literals $\ell^i_a$ and $\ell^j_b$, exactly one of them evaluates to true, then, for two such literals, at most one of the vertices $u^i_a$ and $u^j_b$ is in $S$. This means that, for any two opposite literals $\ell^i_a$ and $\ell^j_b$, at most one of the vertices $w^i_a$ and $w^j_b$ is in $I^6(S) \setminus I^5(S)$, while those that are not in $I^6(S) \setminus I^5(S)$ are in $I^7(S) \setminus I^6(S)$.

This implies that every vertex $y_{i,a,j,b} \in I^8(S) \setminus I^7(S)$, since at most one of its neighbors is in $I^6(S) \setminus I^5(S)$. Therefore, $Y \subseteq I^8(S) \setminus I^7(S)$ and, consequently, $x \in I^9(S) \setminus I^8(S)$, since $Y' \subseteq I^4(S)$. In conclusion, we have that $V(G) = I^9(S)$ and $I^8(S) \neq V(G)$.

Conversely, let $S$ be a hull set of $G$ such that $I^8(S) \neq V(G)$. According to \Cref{lemma:percolationNumberk9pendant}, $\{q^i_1,q^i_2,q^i_3,r^i_1,r^i_2,r^i_3,z^i\} \subseteq I^4(S)$ for each $1 \leq i \leq m$, and~$Y' \subseteq I^4(S)$. Let $U_i=\{u^i_1,u^i_2,u^i_3,v^i_1,v^i_2,v^i_3\}$. For any $1 \leq i \leq m$, $U_i$ is co-convex, since there is no path in $G - U_i$ whose endpoints are both adjacent to the same vertex in $U_i$. This implies that at least one vertex in $U_i$ is also in $S$, for each $1 \leq i \leq m$.

Note that if $S \subseteq S'$ and $v \in I^k(S) \setminus I^{k-1}(S)$, then $v \in I^k(S')$. From the structure of the gadget in \Cref{fig:percolationNumberk9}, we can see that, for any $d \in U_i$, $V_i\subseteq I^7(\{d\})$. This implies that $V_i\subseteq I^7(S)$, since $S \cap U_i \neq \varnothing$. Therefore, $Y \subseteq I^8(S)$ and, consequently, $V(G) \setminus \{x\} \subseteq I^8(S)$. However, since we assume initially that $I^8(S) \neq V(G)$, we conclude that $x \in I^9(S) \setminus I^8(S)$.

Since $x \in I^9(S) \setminus I^8(S)$, $Y \subseteq I^8(S) \setminus I^7(S)$, as any vertex $y_{i,a,j,b} \in I^7(S)$ would imply that $x \in I^8(S)$. Hence, for each $y_{i,a,j,b} \in Y$, at least one of its neighbors, $w^i_a$ and $w^j_b$, is in $I^7(S) \setminus I^6(S)$, that is, we cannot have $w^i_a,w^j_b \in I^6(S)$. 

We then define the following assignment for the variables in $\mathcal{X}$: for each~$w^i_a \in I^6(S)$, we set the variable $X_k$ to true if and only if $\ell^i_a$ is the positive literal of $X_k$. In other words, for each $w^i_a \in I^6(S)$, we set $X_k$ to a value so that $\ell^i_a$ evaluates to true. After making these assignments, if there are still variables remaining without a value assigned, we set them all to true. 

We claim that our assignment is well defined, since we assign a value to all variables and, if it is assigned both true and false to any given variable, then there would be two vertices $w^i_a,w^j_b \in I^6(S)$ such that $\ell^i_a$ and $\ell^j_b$ are opposite literals, which would mean that $y_{i,a,j,b} \in I^7(S)$, a contradiction as~$Y \subseteq I^8(S) \setminus I^7(S)$.

Remember that $S \cap U_i \neq \varnothing$ for each $1 \leq i \leq m$, which means that, for each $1 \leq i \leq m$, there is at least one $1 \leq a \leq 3$ such that $w^i_a \in I^6(S)$. This allows us to assert that there is at least one literal that evaluates to true in each clause, which implies that our assignment satisfies $\varphi$.

To achieve a reduction for any fixed $k > 9$, we simply repeat the following two-step procedure $k-9$ times in the constructed graph: (i) we add one perpetuation gadget along with one additional vertex $t$; and (ii) we add the necessary edges to form a $K_3$ with the vertices $x$, $t$ and the vertex $p$ of the newly added perpetuation gadget, designating $t$ as the new vertex $x$. The proof of equivalence between the instances follows the same reasoning as the one presented for $k = 9$.
\end{proof}

\section{Final Remarks}
\label{sec:finalremarks}

In this paper, we show that \nameref{prob:convexityNumber} is \NP-complete and \W[1]-hard when parameterized by $k$ for thick spiders, but polynomial for extended $P_4$-laden graphs. By demonstrating the \NP-completeness and \W[1]-hardness of \nameref{prob:convexityNumber}, we provide valuable insights into the inherent computational challenges associated with this problem. Moreover, our polynomial-time algorithm for extended $P_4$-laden graphs offers a promising avenue for efficient computation in graphs with few $P_4$s. We leave open for further investigation the complexity of \nameref{prob:convexityNumber} for classes related to the hierarchy of classes with few $P_4$s, such as $(q,q-4)$-graphs, $P_5$-free graphs and $P_6$-free graphs, as well as other well-known classes of graphs such as planar graphs and bipartite graphs.

Regarding the \nameref{prob:percolationNumber} problem, we present a linear algorithm when the input graph is a cactus and a polynomial algorithm for fixed~$k \leq 2$. On the negative side, we show that it is \NP-complete even when~$k$ is fixed at any value at least $9$. The polynomial algorithms for the \nameref{prob:percolationNumber} for cacti and for fixed $k \leq 2$ represent a noteworthy effort to advance our understanding of the tractability of this problem for very restricted instances. On the other hand, the \NP-completeness of \nameref{prob:percolationNumber}, even for fixed values of $k \geq 9$, highlights the inherent difficulty of the problem for more general instances. We leave open for future exploration the complexity of \nameref{prob:percolationNumber} for bipartite graphs and for any fixed value of $k$ between $3$ and $8$. In fact, given the nature of the cycle convexity, which seems to be more complex than the $P_3$ convexity, we conjecture that this problem is \NP-complete for any fixed~$k \geq 4$ in general graphs and for any fixed~$k \geq 5$ in bipartite graphs, as is the case for the $P_3$ convexity \cite{benevides2015,marcilon2018}.

\section*{Acknowledgments}

The statements of some of the results in this paper appeared in ETC-2024 (Encontro de Teoria da Computação). This study was financed in part by the Coordenação de Aperfeiçoamento de Pessoal de Nível Superior (CAPES) -- Finance Code 001 and by the Conselho Nacional de Desenvolvimento Científico e Tecnológico (CNPq) -- Universal [422912/2021-2].

\bibliographystyle{elsarticle-num} 
\bibliography{main}

\end{document}